\documentclass[letterpaper]{article} 
\usepackage{aaai24}  
\usepackage{times}  
\usepackage{helvet}  
\usepackage{courier}  
\usepackage[hyphens]{url}  
\usepackage{graphicx} 
\usepackage{natbib}  
\usepackage{caption} 
\usepackage{algorithm}
\usepackage{algorithmicx}
\usepackage{algpseudocode}

\usepackage{amsthm}
\newtheorem{theorem}{Theorem}

\usepackage{booktabs}
\usepackage{multirow}
\usepackage{tabularx}

\usepackage{subfigure}

\usepackage{amsfonts}
\usepackage{amsmath}
\usepackage{makecell}

\usepackage{newfloat}
\usepackage{listings}
\DeclareCaptionStyle{ruled}{labelfont=normalfont,labelsep=colon,strut=off} 
\floatstyle{ruled}
\newfloat{listing}{tb}{lst}{}
\floatname{listing}{Listing}
\title{Adaptive Hardness Negative Sampling for Collaborative Filtering}
\author{
    Riwei Lai\textsuperscript{\rm 1, 2}, Rui Chen\textsuperscript{\rm 1}\thanks{Corresponding authors}, Qilong Han\textsuperscript{\rm 1}\footnotemark[1], Chi Zhang\textsuperscript{\rm 1}, Li Chen\textsuperscript{\rm 2}\\
}
\affiliations{
    \textsuperscript{1} College of Computer Science and Technology, Harbin Engineering University \\
    \textsuperscript{2} Department of Computer Science, Hong Kong Baptist University\\
    \{lai, ruichen, hanqilong, zhangchi20\}@hrbeu.edu.cn, lichen@comp.hkbu.edu.hk
}

\begin{document}

\maketitle

\begin{abstract}
Negative sampling is essential for implicit collaborative filtering to provide proper negative training signals so as to achieve desirable performance. We experimentally unveil a common limitation of all existing negative sampling methods that they can only select negative samples of a fixed hardness level, leading to the false positive problem (FPP) and false negative problem (FNP). We then propose a new paradigm called adaptive hardness negative sampling (AHNS) and discuss its three key criteria. By adaptively selecting negative samples with appropriate hardnesses during the training process, AHNS can well mitigate the impacts of FPP and FNP. Next, we present a concrete instantiation of AHNS called AHNS$_{p<0}$, and theoretically demonstrate that AHNS$_{p<0}$ can fit the three criteria of AHNS well and achieve a larger lower bound of normalized discounted cumulative gain. Besides, we note that existing negative sampling methods can be regarded as more relaxed cases of AHNS. Finally, we conduct comprehensive experiments, and the results show that AHNS$_{p<0}$ can consistently and substantially outperform several state-of-the-art competitors on multiple datasets.
\end{abstract}

\section{Introduction}
Collaborative filtering (CF), as the most representative technique for recommendation, focuses on modeling user interests from observed user-item interactions~\cite{WHW19, HDK20}. In many cases, it is not always possible to obtain a large amount of high-quality explicit feedback. As a result, implicit feedback, such as clicks or purchases, has become a default choice to train a CF model~\cite{LCZ23}. In implicit feedback, each observed interaction normally indicates a user's interest in an item and corresponds to a positive training sample. As for negative training samples, a widely adopted approach is to randomly select some uninteracted items for users. An implicit CF model is then optimized to give positive samples higher scores than negative ones~\cite{RFG09}.

Similar to many semi-supervised learning problems, existing implicit CF models highly rely on mining negative samples to provide proper negative training signals. Without auxiliary data describing items, two lines of works have been proposed. The first line consists of \emph{static negative sampling}, which assigns a static probability for each candidate to be sampled. For example, random negative sampling (RNS)~\cite{RFG09} chooses uninteracted items with equal probability, and popularity-biased negative sampling (PNS)~\cite{CSS17, WVS19} adopts item-popularity-biased distributions to favor popular items. The other line is \emph{hard negative sampling}, such as dynamic negative sampling (DNS)~\cite{ZCW13} and disentangled negative sampling (DENS)~\cite{LCZ23}, which focuses on selecting hard negative samples that are difficult to be distinguished from the positive samples with dynamic distributions. Such hard negative samples can provide more informative training signals so that user interests can be better characterized~\cite{XLZ22}.

\begin{figure}[t]
  \centering
  \includegraphics[width=0.9\linewidth]{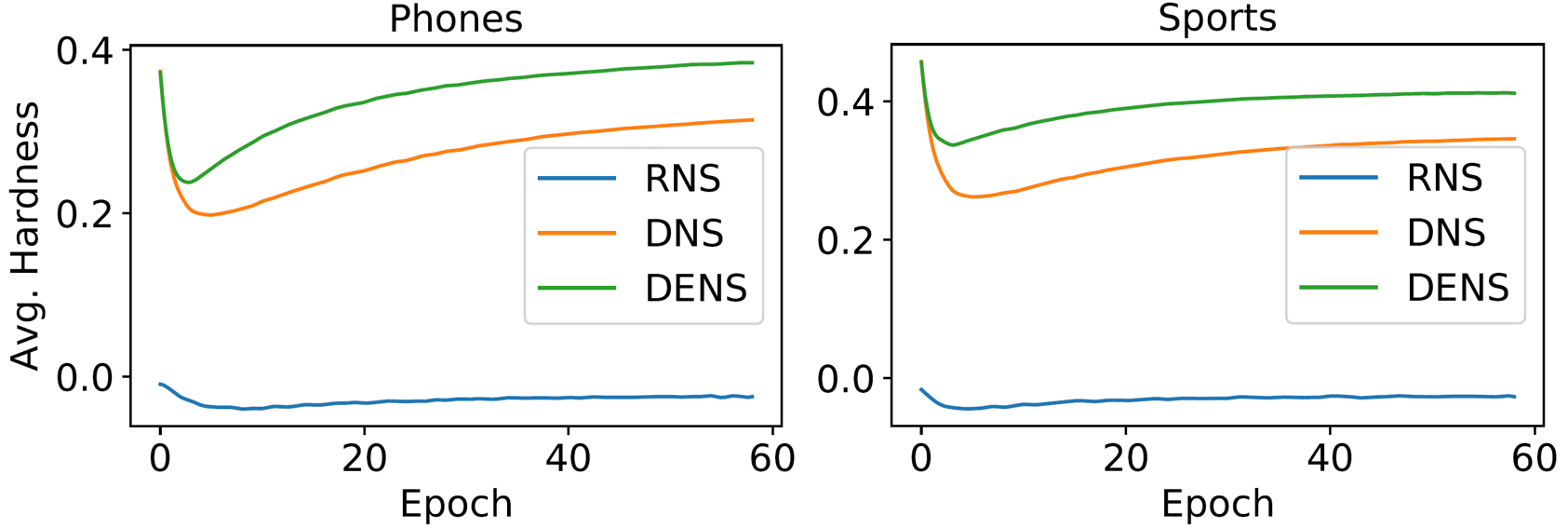}
  \caption{Average hardness of selected negative items in RNS, DNS, and DENS on two Amazon datasets.}
  \label{fig:moti}
\end{figure}

Although the above two lines of works on negative sampling have achieved some promising results, we point out that all these methods can only select negative samples of a certain ``hardness'' level, preventing them from achieving better performance. Without loss of generality, assume that positive samples' predicted scores are always positive. We can define the \textit{hardness} of a negative sample as its relative predicted score, i.e., the ratio of its predicted score to that of its corresponding positive sample, in order to smooth the influence of the simultaneous increase in the predicted scores of all items during the training process. As illustrated in Fig.~\ref{fig:moti}, throughout the training process, RNS 
can only select easy negative samples with hardness around 0, while DNS 
and DENS 
can only choose hard negative samples with hardness around 0.3 and 0.4, respectively.

Unavoidably, these fixed hardness negative sampling methods may suffer from two significant problems: (1) \emph{false positive problem} (FPP): as shown in the upper part of Fig.~\ref{fig:hfns}, when only easy negative samples can be selected during the training process, items of no interest but with initially high predicted scores may not be sufficiently updated and will still be recommended to users, resulting in suboptimal recommendation results; (2) \emph{false negative problem} (FNP): as shown in the lower part of Fig.~\ref{fig:hfns}, if only hard negative samples with a fixed hardness level are selected during the training process, items of interests but not interacted yet may be selected as negative and ranked lower in the recommendation list, which worsens recommendation results. We have conducted extensive experiments to verify the existence of FPP and FNP (see RQ2 of Experiments for more details).

\begin{figure}[t]
  \centering
  \includegraphics[width=0.75\linewidth]{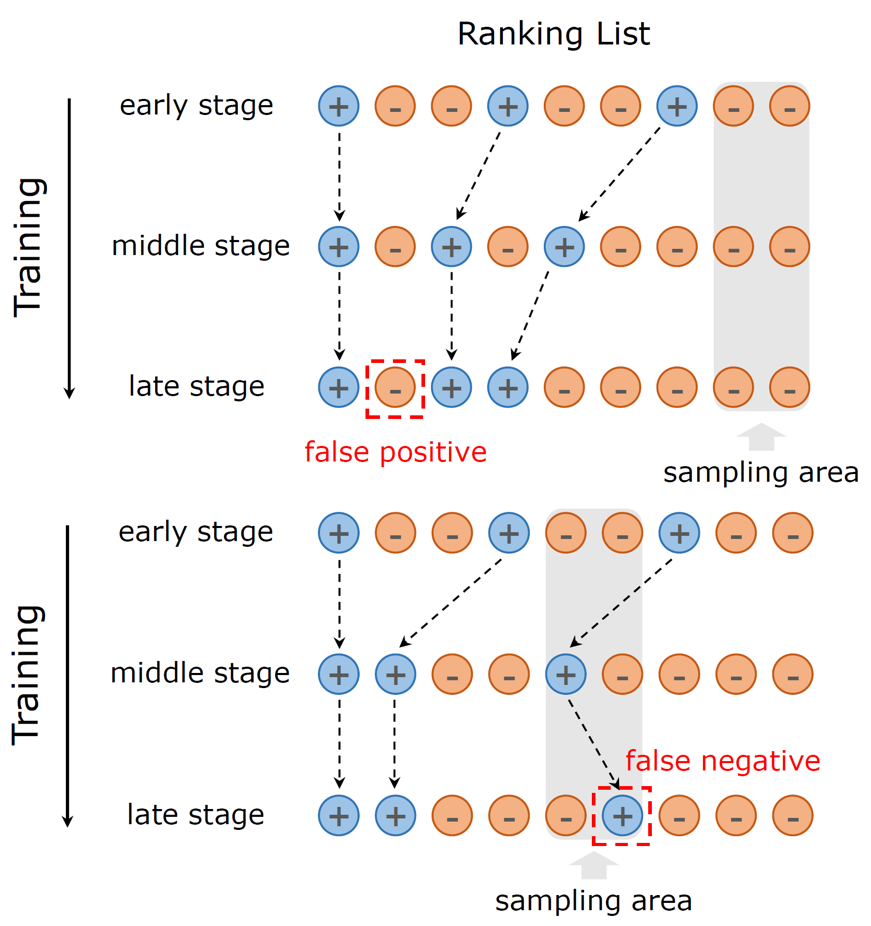}
  \caption{Issues of fixed hardness negative sampling.}
  \label{fig:hfns}
\end{figure}

To address the above two problems and obtain better recommendation results, we propose to adaptively select negative samples with different hardness levels during the training process. A straightforward attempt is to introduce curriculum learning~\cite{CWS21} into negative sampling, where a predefined pacing function is utilized to schedule the hardness levels of negative samples in different training epochs. However, such an implementation still selects negative samples with a fixed hardness level within the same epoch, rather than \textit{adaptively select negative samples with different hardness levels for different positive samples}.

In this paper, we introduce a brand new negative sampling paradigm called \underline{A}daptive \underline{H}ardness \underline{N}egative \underline{S}ampling (AHNS), and analyze its three key criteria. We then present a concrete instantiation of AHNS called AHNS$_{p<0}$, where $p$ is a predefined smoothing parameter we will explain later. Comprehensive theoretical analyses are performed to confirm that AHNS$_{p<0}$ satisfies the three criteria of AHNS and prove that implicit CF models with AHNS$_{p<0}$ can achieve a larger lower bound on normalized discounted cumulative gain (NDCG) than with a fixed hardness negative sampling method. Furthermore, we discuss the relation between AHNS and other negative sampling methods and note that existing negative sampling methods can be considered as more relaxed cases of AHNS. 

Our main contributions are summarized as follows:

\begin{itemize}
    \item We are the first to identify and address FPP and FNP in existing negative sampling methods via adaptively selecting hardnesses of negative samples, which brings a new perspective of negative sampling for implicit CF.
    \item We propose a new negative sampling paradigm AHNS with three criteria, which generalizes existing negative sampling methods. We present a concrete instantiation AHNS$_{p<0}$ and theoretically show that it can fit the three criteria well and achieve a larger lower bound on NDCG.
    \item We conduct extensive experiments to demonstrate that AHNS$_{p<0}$ can achieve significant improvements over several representative state-of-the-art negative sampling methods.
\end{itemize}

\section{Related Work}

\subsection{Static Negative Sampling}
Static negative sampling focuses on identifying good distributions to draw negative samples. For example, as the simplest and most prevalent static negative sampling method, Bayesian personalized ranking (BPR)~\cite{RFG09} randomly selects uninteracted items as negative. However, this method makes it hard to guarantee the quality of selected negative samples, and thus some studies~\cite{CSS17, WVS19, YDZ20} propose to replace the uniform distribution with other distributions. Inspired by the word-frequency-based and node-degree-based negative sampling distributions for network embedding~\cite{MSC13}, NNCF~\cite{CSS17} and NCE-PLRec~\cite{WVS19} adopt an item-popularity-based sampling distribution to select more popular items as negative, which helps to alleviate the widespread popularity bias issue in recommender systems~\cite{CDW23}.

\subsection{Hard Negative Sampling}
Hard negative sampling methods emphasize the importance of oversampling hard negative samples to speed up the training process and find more precise delineations of user interests. More specifically, it is achieved by either assigning higher sampling probabilities to items with larger predicted scores~\cite{ZCW13, DQY20, HDD21, ZZH22, LCZ23, SCF23, ZCL23} or leveraging adversarial learning techniques~\cite{WYZ17, CW18, PC19}. For instance, dynamic negative sampling (DNS)~\cite{ZCW13} selects the item with the highest predicted score in a candidate negative sample set. SRNS~\cite{DQY20} oversamples items with both high predicted scores and high variances to tackle the false negative problem. DENS~\cite{LCZ23} disentangles relevant and irrelevant factors of items and identifies the best negative samples with a factor-aware sampling strategy. Instead of directly selecting negative samples from uninteracted items, MixGCF~\cite{HDD21} synthesizes hard negative samples by mixing positive information into negative samples, which further improves the performance. 

However, we experimentally find that all the above negative sampling methods can only select negative samples of a fixed hardness level during the training process, leading to the false positive problem and false negative problem. Driven by this limitation, we propose an adaptive hardness negative sampling paradigm, which adaptively selects negative samples with appropriate hardnesses and achieves better recommendation results.

\section{Proposed Method}

\subsection{Problem Formulation}
In this section, we formulate the problem of negative sampling in implicit CF. Let $\mathcal{U}$ and $\mathcal{I}$ be the set of users and the set of items, respectively. We denote the set of observed interactions, i.e., implicit feedback, by $\mathcal{O}^+ = \{(u, i^+) \mid u \in \mathcal{U}, i^+ \in \mathcal{I}\}$, where each pair $(u, i^+)$ indicates an interaction between user $u$ and item $i^+$. Implicit CF aims to characterize user interests from their observed interactions. Interacted items are generally used to form positive pairs, while uninteracted items are considered candidates to generate negative samples. Specifically, given a positive pair $(u, i^+)$, a negative sampling strategy identifies an item $i^-$ that has not been previously interacted by $u$ as a negative sample. The implicit CF model is then optimized to give positive pairs higher scores than negative pairs by the Bayesian personalized ranking (BPR) loss function~\cite{RFG09}:

\begin{equation}
    \mathcal{L}_{\mathrm{BPR}} = \sum_{(u, i^+, i^-)}  - \ln \sigma (\mathbf{e}_u^\top \mathbf{e}_{i^+} - \mathbf{e}_u^\top \mathbf{e}_{i^-}),
\end{equation}

\noindent where $\mathbf{e}_u$, $\mathbf{e}_{i^+}$, and $\mathbf{e}_{i^-}$ are the embeddings of user $u$, positive sample $i^+$, and negative sample $i^-$, respectively, the inner product is used to measure the score of positive and negative pairs, and $\sigma (\cdot)$ is the sigmoid function.

\subsection{Method Design}

\begin{figure}[t]
  \centering
  \includegraphics[width=0.8\linewidth]{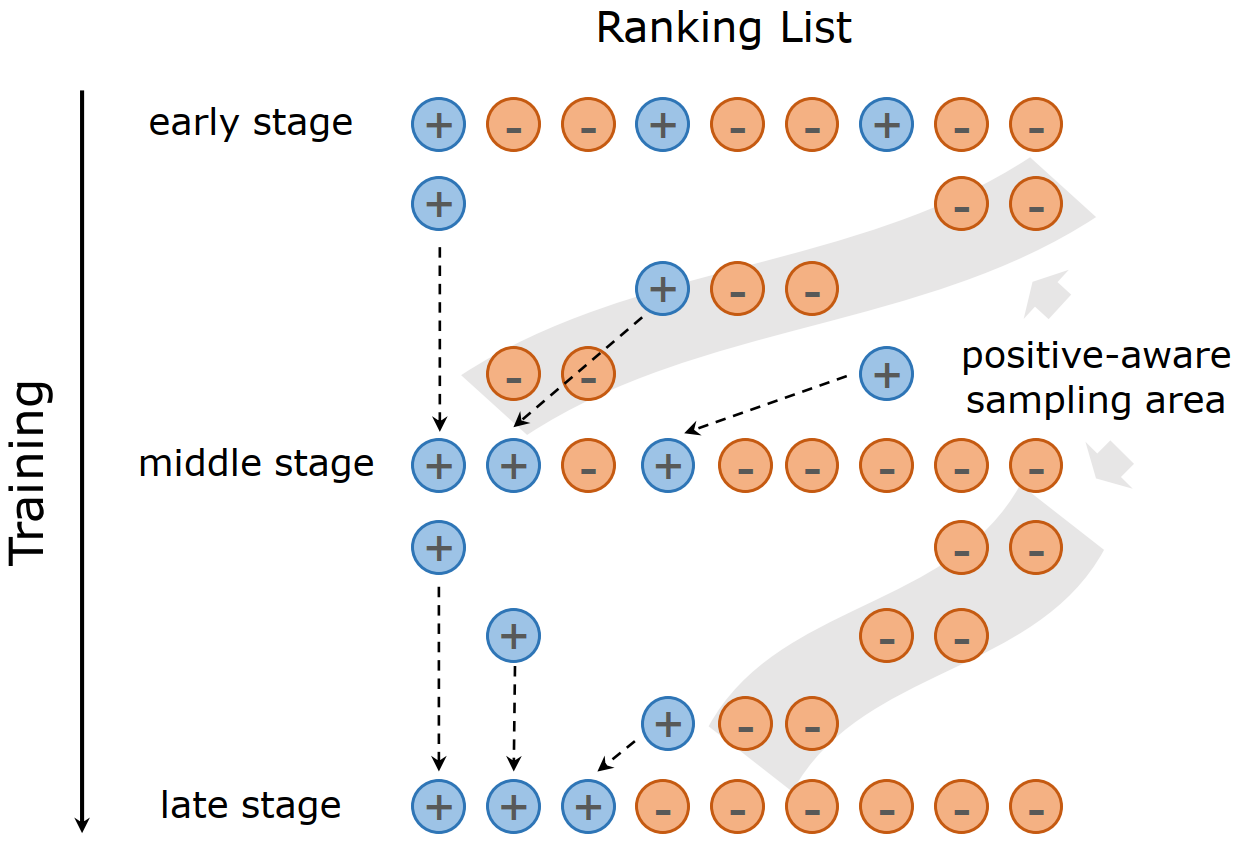}
  \caption{An illustration of adaptive hardness negative sampling.}
  \label{fig:ahns}
\end{figure}

\subsubsection{Paradigm.} 
To achieve the adaptive selection in the hardnesses of negative samples and alleviate the false positive problem (FPP) and false negative problem (FNP), we propose the adaptive hardness negative sampling (AHNS) paradigm. As shown in Fig.~\ref{fig:ahns}, unlike fixed hardness negative sampling, AHNS simultaneously satisfies the following three key criteria:
\begin{itemize}
    \item \textbf{C1: The hardness of a selected negative sample should be positive-aware.} Instead of setting a specific hardness level of negative samples for each training epoch like curriculum learning~\cite{CWS21}, AHNS is expected to identify the appropriate hardness of a negative sample according to its corresponding positive sample.
    \item \textbf{C2: The hardness of a selected negative sample should be negatively correlated with the predicted score of its corresponding positive sample.} On the one hand, for positive samples with higher predicted scores, AHNS should select items with lower hardnesses as negative, which can effectively avoid the FNP. On the other hand, for positive samples with lower predicted scores, AHNS should select items with higher hardnesses as negative, which can accelerate the optimization of positives and enable negatives with higher hardnesses to be sufficiently updated, thus alleviating the FPP.
    \item \textbf{C3: The hardness of selected negative samples should be adjustable.} To cover a variety of practical recommendation scenarios, e.g., different datasets or evaluation metrics~\cite{SCF23}, AHNS should be able to adjust the hardness of selected negative samples.
\end{itemize}

\begin{algorithm}[t]
\caption{AHNS$_{p<0}$}
\label{alg:hans}
\begin{algorithmic}[1]
\State \textbf{Input:} Set of observed interactions $\mathcal{O}^+ = \{(u, i^+) \mid u \in \mathcal{U}, i^+ \in \mathcal{I}\}$, number of candidate negatives $M$, predefined hyperparameters $\alpha$, $\beta$ and $p$
\State \textbf{Output:} Set of training triples $\mathcal{T}$
\State $\mathcal{T} \gets \{\}$ \Comment{Initialize an empty set for training triples}

\For{each positive pair $(u, i^+)$ in $\mathcal{O}^+$}
    \State $\mathcal{C} \gets \{\}$ \Comment{Initialize an empty set for candidate negative samples}
    \For{$m = 1$ to $M$}
        \State $i_m \gets$ Randomly sample an uninteracted item
        \State Add $i_m$ to $\mathcal{C}$
    \EndFor
    \State $\mathcal{R} \gets \{\}$ \Comment{Initialize an empty set for ratings of candidate negative samples}
    \For{each candidate negative sample $i_m$ in $\mathcal{C}$}
        \State $r_m \gets$ $\left\vert \mathbf{e}_u^\top \mathbf{e}_{i_m} - \beta \cdot (\mathbf{e}_u^\top \mathbf{e}_{i^+} + \alpha)^{p+1} \right\vert$
        \State Add $r_m$ to $\mathcal{R}$
    \EndFor
    \State $i^- \gets$ Select $i_m$ with the smallest $r_m$ in $\mathcal{R}$
    \State Add $(u, i^+, i^-)$ to $\mathcal{T}$
\EndFor
\State \Return $\mathcal{T}$
\end{algorithmic}
\end{algorithm}

\subsubsection{Instantiation.}
Next, we give a concrete instantiation of AHNS called AHNS$_{p<0}$, whose entire procedure is detailed in Algo.~\ref{alg:hans}. Specifically, for a positive pair $(u, i^+)$, we follow conventional methods~\cite{CLJ22, LCZ23} and adapt the two-pass sampling idea, which first randomly samples a fixed size of uninteracted items to form a candidate set, and then selects a negative sample from the candidate set according to predefined rating functions and sampling rules. For the first pass, the size $M$ of the candidate set $\mathcal{C}$ is usually much smaller than the total number of items $\left\vert \mathcal{I} \right \vert$, which can boost the sampling efficiency. For the second pass, the rating function and sampling rule play a critical role in identifying the final negative sample and are the focus of all negative sampling methods. Therefore, we introduce three hyperparameters and carefully design the rating function in AHNS$_{p<0}$. For each candidate negative item $i_m \in \mathcal{C}$, the rating function is formulated as:
\begin{equation}
    r_m = \left\vert \mathbf{e}_u^\top \mathbf{e}_{i_m} - \beta \cdot (\mathbf{e}_u^\top \mathbf{e}_{i^+} + \alpha)^{p+1} \right\vert,
\label{eq:rm}
\end{equation}
\noindent where $\alpha > 0$, $\beta > 0$ and $p < 0$ are predefined hyperparameters, whose effects will be given in the subsequent Thm.~\ref{thm:c3}. After calculating the ratings of all candidate negative items, we obtain a rating set $\mathcal{R}$, and then the final negative sample is identified by selecting $i_m$ with the smallest $r_m$ in $\mathcal{R}$:
\begin{equation}
    i^- = i_{\mathop{\arg\min}_m r_m}.
\label{eq:neg}
\end{equation}

\subsection{Theoretical Analysis}
In this section, we conduct in-depth analyses on AHNS$_{p<0}$. We first show that AHNS$_{p<0}$ satisfies the three criteria of AHNS, and then establish that implicit CF models with AHNS$_{p<0}$ can achieve a larger lower bound on normalized discounted cumulative gain than with a fixed hardness negative sampling method as training progresses. 

\begin{theorem}
AHNS$_{p<0}$ satisfies C2 of AHNS.
\label{thm:c2}
\end{theorem}

\begin{proof}
Consider a positive pair $(u, i^+)$. Let $i_*^-$ be the ideal negative sample selected by AHNS$_{p<0}$. According to Eq.~(\ref{eq:rm}) and Eq.~(\ref{eq:neg}), we have: 
\begin{equation}
    \mathbf{e}_u^\top \mathbf{e}_{i_*^-}  = \beta \cdot (\mathbf{e}_u^\top \mathbf{e}_{i^+} + \alpha)^{p+1}.
\end{equation}
To simplify the calculation process, we substitute $\mathbf{e}_u^\top \mathbf{e}_{i^+}$ with $(\mathbf{e}_u^\top \mathbf{e}_{i^+}  + \alpha)$ to calculate the hardness of $i_*^-$:
\begin{align}
    \mathrm{Hardness}(i_*^-) = \frac{\mathbf{e}_u^\top \mathbf{e}_{i_*^-}}{\mathbf{e}_u^\top \mathbf{e}_{i^+}}  &\approx \frac{\mathbf{e}_u^\top \mathbf{e}_{i_*^-}}{\mathbf{e}_u^\top \mathbf{e}_{i^+} + \alpha} \notag \\
    & = \frac{\beta \cdot (\mathbf{e}_u^\top \mathbf{e}_{i^+} + \alpha)^{p+1}}{\mathbf{e}_u^\top \mathbf{e}_{i^+} + \alpha} \notag \\
    & = \beta \cdot (\mathbf{e}_u^\top \mathbf{e}_{i^+} + \alpha)^p.
\label{eq:hard}
\end{align}
Based on the chain rule, we have:
\begin{align}
    \frac{d \mathrm{Hardness}(i_*^-)}{d (\mathbf{e}_u^\top \mathbf{e}_{i^+})} &= \frac{d (\beta \cdot (\mathbf{e}_u^\top \mathbf{e}_{i^+} + \alpha)^p)}{d (\mathbf{e}_u^\top \mathbf{e}_{i^+} + \alpha)} \cdot \frac{d (\mathbf{e}_u^\top \mathbf{e}_{i^+} + \alpha)}{d (\mathbf{e}_u^\top \mathbf{e}_{i^+})} \notag \\
    &= p \cdot \beta \cdot (\mathbf{e}_u^\top \mathbf{e}_{i^+} + \alpha)^{p-1}.
\end{align}
Clearly, $d \mathrm{Hardness}(i_*^-) / d (\mathbf{e}_u^\top \mathbf{e}_{i^+}) < 0$ always holds when $\mathbf{e}_u^\top \mathbf{e}_{i^+} > 0$, $\alpha > 0$, $\beta > 0$ and $p < 0$, which means that the hardness of $i_*^-$ is always negatively correlated with the predicted score of $i^+$ -- the above completes the proof.
\end{proof}

\begin{theorem}
AHNS$_{p<0}$ satisfies C1 of AHNS.
\end{theorem}

\begin{proof}
Consider two different positive pairs $(u, i^+_1)$ and $(u, i^+_2)$. Let $i_{1*}^-$ and $i_{2*}^-$ be the ideal negative samples selected by AHNS$_{p<0}$ for $(u, i^+_1)$ and $(u, i^+_2)$, respectively. According to Eq.~(\ref{eq:hard}), we have: 
\begin{align}
    \mathrm{Hardness}(i_{1*}^-)  &= \beta \cdot (\mathbf{e}_u^\top \mathbf{e}_{i^+_1} + \alpha)^p, \notag \\
    \mathrm{Hardness}(i_{2*}^-)  &= \beta \cdot (\mathbf{e}_u^\top \mathbf{e}_{i^+_2} + \alpha)^p.
\end{align}
It has been proved in Thm.~\ref{thm:c2} that $\mathrm{Hardness}(i_*^-)$ monotonically decreases as $\mathbf{e}_u^\top \mathbf{e}_{i^+}$ increases. Thus when $\mathbf{e}_u^\top \mathbf{e}_{i^+_1} \ne \mathbf{e}_u^\top \mathbf{e}_{i^+_2}$, $\mathrm{Hardness}(i_{1*}^-) \ne \mathrm{Hardness}(i_{2*}^-)$ -- the above completes the proof.
\end{proof}

\begin{figure}[t]
  \centering
  \includegraphics[width=0.75\linewidth]{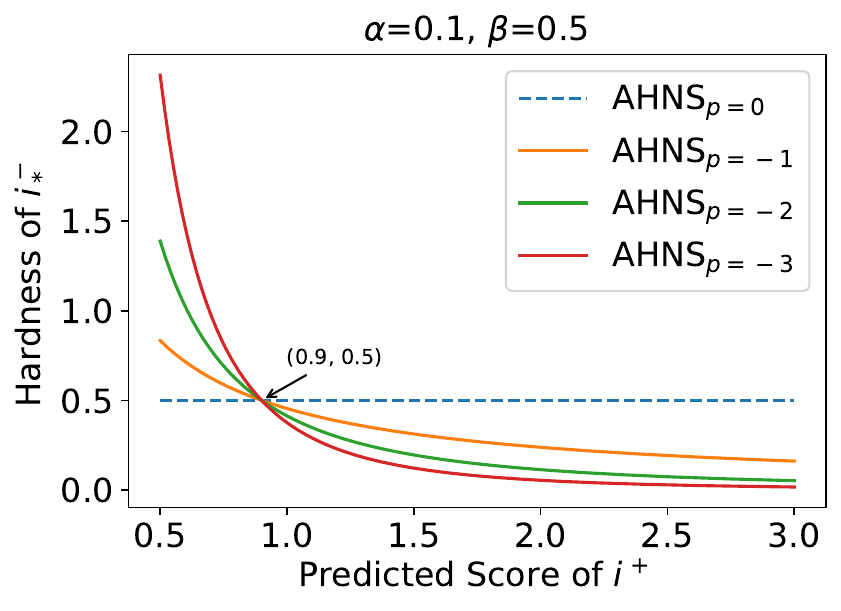}
  \caption{Hardness of ideal negative sample $i_*^-$ \textit{w.r.t.} different $p$.}
  \label{fig:ideal}
\end{figure}

\begin{theorem}
AHNS$_{p<0}$ satisfies C3 of AHNS.
\label{thm:c3}
\end{theorem}

\begin{proof}
According to Eq.~(\ref{eq:hard}), we plot the curves of the hardness of $i_*^-$ under different values of the predicted score of $i^+$ in Fig.~\ref{fig:ideal}. It is clear that $p$ affects the magnitude of the curves, smaller $p$ leads to larger magnitudes. In addition, all curves pass through the point $(1 - \alpha, \beta)$, indicating the effect of $\alpha$ and $\beta$ in adjusting the hardness of selected negative samples--the above completes the proof.
\end{proof}

\begin{theorem}
As training progresses, implicit CF models with AHNS$_{p<0}$ can achieve a larger lower bound on normalized discounted cumulative gain (NDCG) than with a fixed hardness negative sampling method.
\end{theorem}

\begin{proof}
Given a user $u$, let $\pi_{f_u}$ be the ranking function induced by recommender system $f$ for user $u$, and $\pi_{f_u}(i)$ the rank of item $i$. Let $y$ be a binary indicator: $y_i = 1$ if item $i$ has been interacted by $u$, otherwise $y_i = 0$. Let $\mathcal{I}(u) = \{i \mid y_i = 1\}$ be the set of items interacted by $u$ and $\mathbb{I}$ be the indicator function.

First, we consider discounted cumulative gain (DCG). With $1 + z \le 2^z$ when $z \geq 1$, we have the following:
\begin{align}
    DCG(u) &= \sum_{i = 1}^{\left\vert \mathcal{I} \right\vert} \frac{2^{y_i} - 1}{\log_2(1 + \pi_{f_u}(i))} \notag \\ 
     & = \sum_{i = 1}^{\left\vert \mathcal{I}(u) \right\vert} \frac{1}{\log_2(1 + \pi_{f_u}(i))} \geq \sum_{i = 1}^{\left\vert \mathcal{I}(u) \right\vert} \frac{1}{\pi_{f_u}(i)} \notag \\
    & = \sum_{i = 1}^{\left\vert \mathcal{I}(u) \right\vert} \frac{1}{1 + \sum_{j \in \left\vert \mathcal{I} \right\vert \setminus \{i\}} \mathbb{I}(\mathbf{e}_u^\top \mathbf{e}_j  - \mathbf{e}_u^\top \mathbf{e}_i > 0)} \notag \\
    & \geq \sum_{i = 1}^{\left\vert \mathcal{I}(u) \right\vert} \frac{1}{1 + \sum_{j \in \left\vert \mathcal{I} \right\vert \setminus \{i\}} \exp(\mathbf{e}_u^\top \mathbf{e}_j  - \mathbf{e}_u^\top \mathbf{e}_i)}.
\end{align}

Next, we consider the ideal DCG (IDCG). Let $\pi_{f_u}^*$ be the ideal ranking function, which can sort the items in the ground truth order:
\begin{align}
    IDCG(u) &= \sum_{i = 1}^{\left\vert \mathcal{I} \right\vert} \frac{2^{y_i} - 1}{\log_2(1 + \pi_{f_u}^*(i))} \notag \\
    & = \sum_{i = 1}^{\left\vert \mathcal{I}(u) \right\vert} \frac{1}{\log_2(1 + i)} \le \sum_{i = 1}^{\left\vert \mathcal{I}(u) \right\vert} 1 = \left\vert \mathcal{I}(u) \right\vert.
\end{align}
Clearly, we have:
\begin{align}
    \frac{1}{IDCG(u)} \geq \frac{1}{\left\vert \mathcal{I}(u) \right\vert}.
\end{align}

Finally, we consider NDCG:
\begin{align}
    NDCG&(u) = \frac{DCG(u)}{IDCG(u)} \geq \frac{1}{\left\vert \mathcal{I}(u) \right\vert} DCG(u) \notag \\ 
    & \geq \frac{1}{\left\vert \mathcal{I}(u) \right\vert} \sum_{i = 1}^{\left\vert \mathcal{I}(u) \right\vert} \frac{1}{1 + \sum_{j \in \left\vert \mathcal{I} \right\vert \setminus \{i\}} \exp(\mathbf{e}_u^\top \mathbf{e}_j  - \mathbf{e}_u^\top \mathbf{e}_i)}  \notag \\
    & \approx \frac{1}{\left\vert \mathcal{I}(u) \right\vert} \sum_{i = 1}^{\left\vert \mathcal{I}(u) \right\vert} \frac{1}{1 + \exp(\mathbf{e}_u^\top \mathbf{e}_{i^-} - \mathbf{e}_u^\top \mathbf{e}_i)}.
\end{align}

As illustrated in Fig.~\ref{fig:ideal}, it is not difficult to derive that as the predicted score of $i^+$ increases, the hardness of $i_*^-$ sampled by AHNS$_{p<0}$ (the solid lines) is lower compared to that of a fixed hardness negative sampling method (the dotted line), leading to a lower value of $\exp(\mathbf{e}_u^\top \mathbf{e}_{i^-} - \mathbf{e}_u^\top \mathbf{e}_i)$. Thus implicit CF models with AHNS$_{p<0}$ can achieve a larger lower bound of NDCG. This completes the proof. 
\end{proof}

\subsection{Discussion}
In this section, we first discuss the relation between AHNS and other negative sampling methods. We point out that existing negative sampling methods can be considered as more relaxed cases that satisfy part of the three criteria of AHNS. For example, DENS~\cite{LCZ23} proposes a positive gating layer to disentangle items' factors for negative sampling. Thus the hardness of selected negative samples becomes positive-aware and satisfies C1 of AHNS. By using an anti-curriculum pacing function to schedule the hardnesses of negative samples for different training epochs, CuCo~\cite{CWS21} partially satisfies C2 of AHNS. To adapt to different datasets and top-$K$ metrics, DNS($M, N$)~\cite{SCF23} adjusts the hardnesses of selected negative samples via predefined hyperparameters, which satisfies C3 of AHNS.

In addition, we note that the main idea of AHNS in negative sampling is consistent with that of focal loss~\cite{LGG17} in object detection, i.e., \textbf{putting more focus on lower-ranked positives and higher-ranked negatives (hard, misclassified examples)}, which may bring some new insights into negative sampling for implicit CF. 

\section{Experiments}
In this section, we perform extensive experiments to evaluate AHNS$_{p<0}$ and answer the following research questions:
\begin{itemize}
    \item \textbf{RQ1:} How does AHNS$_{p<0}$ perform compared with previous negative sampling methods?
    \item \textbf{RQ2:} Does AHNS$_{p<0}$ achieve adaptive selection in the hardnesses of negative samples and alleviate the false positive problem (FPP) and false negative problem (FNP)?
    \item \textbf{RQ3:} What are the impacts of the hyperparameters (e.g., $\alpha$, $\beta$) on AHNS$_{p<0}$?
    \item \textbf{RQ4:} Does AHNS$_{p<0}$ have an advantage in terms of sampling efficiency?
\end{itemize}

\begin{table}[t]
\centering
\begin{tabular}{@{}c|c|c|c|c|c@{}}
\toprule
Dataset &  \makecell[c]{\#user \\ $(\left\vert \mathcal{U} \right \vert)$} & \makecell[c]{\#item \\ $(\left\vert \mathcal{I} \right \vert)$} & \makecell[c]{\#inter. \\ $(\left\vert \mathcal{R} \right \vert)$} & \makecell[c]{avg. inter. \\ per user} & density\\ \midrule
ML-1M & 6.0k & 3.7k & 1000.2k & 165.6 &4.47\%\\ 
Phones & 27.9k & 10.4k & 194.4k & 7.0 &0.07\%\\
Sports & 35.6k & 18.4k & 296.3k & 8.3 &0.05\%\\
Tools & 16.6k & 10.2k & 134.5k & 8.1 &0.08\%\\
\bottomrule
\end{tabular}
\caption{The statistics of four datasets.}
\label{tab:data}
\end{table}

\begin{table*}[t]
\centering
\begin{tabular}{@{}c|ccc|ccc|ccc|ccc@{}}
\toprule
\multirow{2}{*}{Method} & \multicolumn{3}{c|}{ML-1M} & \multicolumn{3}{c|}{Phones} & \multicolumn{3}{c|}{Sports} & \multicolumn{3}{c}{Tools} \\ \cmidrule(l){2-13}  & \textit{R}@20    & \textit{N}@20   & \textit{N}@50   & \textit{R}@20    & \textit{N}@20    & \textit{N}@50   & \textit{R}@20    & \textit{N}@20    & \textit{N}@50    & \textit{R}@20    & \textit{N}@20   & \textit{N}@50   \\ \midrule
RNS & 22.86 & 35.46 & 37.41 & 11.06 & 5.98 & 7.35 & 6.73 & 3.60 & 4.68 & 5.53 & 2.99 & 3.75\\
SSM & 24.87 & \underline{37.74} & \underline{39.71} & 11.37 & 6.13 & 7.48  & 7.08 & 3.80 & 4.87 & 5.72 & 3.10 & 3.88\\ 
DNS & 24.66 & 36.64 & 38.31 & 12.08 & 6.64 & 7.99 & 7.74 & 4.25 & 5.32 & 6.66 & 3.78 & 4.52\\
MixGCF & 24.75 & 37.54 & 38.95 & 12.20 & 6.73 & 8.13 & 7.68 & 4.32 & 5.36 & 6.82 & 3.88 & 4.59 \\
DENS & 25.07 & 37.67 & 39.11 & 12.16 & 6.68 & 8.13 & 7.90 & 4.35 & 5.50 & 6.66 & 3.76 & 4.55 \\
DNS(\textit{M},~\textit{N}) & 25.09 & 37.58 & 39.22 & 12.27 & 6.75 & 8.15 & 7.84 & 4.31 & 5.35 & 6.86 & 3.76 & 4.61 \\
CuCo  & 25.12 & 37.53 & 39.20 & 12.19 & 6.68 & 8.11 & 7.68 & 4.25 & 5.36 & 6.76 & 3.82 & 4.59 \\ \midrule
AHNS$_{p=-1}$ & \underline{25.17} & 37.72 & 39.31 & \underline{13.02} & \underline{7.08} & \underline{8.71} & \underline{8.42} & \underline{4.58} & \textbf{5.82} & \underline{7.27} &\underline{4.02} & \textbf{4.95}\\
AHNS$_{p=-2}$ & \textbf{25.51} & \textbf{38.77} & \textbf{40.57} & \textbf{13.03} & \textbf{7.14} & \textbf{8.74} & \textbf{8.52} & \textbf{4.61} & \underline{5.81} & \textbf{7.42} & \textbf{4.05} & \underline{4.92} \\ \midrule
\textit{Improv.} & 1.6\% & 2.7\% & 2.2\% & 6.2\% & 5.8\% & 7.2\% & 7.8\% & 6.0\% & 5.8\% & 8.2\% & 4.4\% & 7.4\%\\ \bottomrule
\end{tabular}
\caption{Performances (\%) of AHNS$_{p=-1}$, AHNS$_{p=-2}$, and baseline methods. The best results are in bold, and the second best are underlined. Improvements are calculated over the best baseline method and statistically significant with $p\text{-value} < 0.01$.}
\label{tab:per}
\end{table*}

\subsection{Experimental Setup}
\subsubsection{Datasets and Evaluation Metrics.}
We consider four widely used public benchmark datasets in experiments: MovieLens-1M\footnote{https://grouplens.org/datasets/movielens/} (ML-1M), Amazon-Phones\textsuperscript{2} (Phones),  Amazon-Sports\textsuperscript{2} (Sports) and Amazon-Tools\footnote{https://jmcauley.ucsd.edu/data/amazon/} (Tools). Following~\cite{HDK20, SCF23}, we randomly split each user’s interactions into training/test sets with a ratio of 80\%/20\%, and build the validation set by randomly sampling 10\% interactions of the training set. Tab.~\ref{tab:data} summarizes the statistics of the four datasets. We report the recommendation performances in terms of $Recall$@20 ($R$@20) and $NDCG$@\{20, 50\} ($N$@\{20, 50\}), where higher values indicate better performances.

\subsubsection{Baseline Methods.}
We compare AHNS$_{p<0}$ with a wide range of representative negative sampling methods:
\begin{itemize}
    \item \textbf{RNS}~\cite{RFG09} randomly selects uninteracted items as negative.
    \item \textbf{SSM}~\cite{WWG22} achieves better performances by sampling more items as negative. 
    \item \textbf{DNS}~\cite{ZCW13} chooses the item with the highest predicted score in a candidate set as negative.
    \item \textbf{MixGCF}~\cite{HDD21} synthesizes harder negative samples by injecting information from positive samples.
    \item \textbf{DENS}~\cite{LCZ23} identifies better negative samples by disentangling factors of items.
    \item \textbf{DNS($M, N$)}~\cite{SCF23} controls the sampling hardness via predefined hyperparameters.
    \item \textbf{GuCo}~\cite{CWS21} proposes a negative sampling method adopting curriculum learning in graph representation learning. We transfer this method to CF.
\end{itemize}

\subsubsection{Implementation Details.}
We strictly follow the experimental setting in DENS~\cite{LCZ23}. We utilize matrix factorization (MF) as the implicit CF model. The embedding dimension is fixed to 64, and the embedding parameters are initialized with the Xavier method. We optimize all parameters with Adam~\cite{KB15} and use the default learning rate of 0.001 and default mini-batch size of 2,048. The number of training epochs is set to 100. For AHNS$_{p<0}$, the candidate negative size $M$ is searched in the range of $\{4, 8, 16, 32, 64\}$. The hyperparameters $\alpha$ and 
$\beta$ are tuned over $\{0.1, 0.2, \cdots, 0.9, 1.0\}$ independently. The hyperparameters of all baseline methods are carefully tuned by grid search. Our code is publically available at {\url{https://github.com/Riwei-HEU/AHNS}}.

\subsection{RQ1: Performance Comparison}

Tab.~\ref{tab:per} shows the performances of AHNS$_{p=-1}$, AHNS$_{p=-2}$, and baseline methods. We can observe the following:
\begin{itemize}
    \item Compared to randomly selecting uninteracted items as negative (RNS), increasing the number (SSM) or the hardness (DNS, MixGCF, DENS, etc.) of negative samples leads to a substantial performance improvement.
    \item By introducing curriculum learning into negative sampling, CuCo draws negative samples with different hardnesses in different training epochs, achieving comparable performances to hard negative sampling methods.
    \item Benefiting from positive-aware adaptive selection in the hardnesses of negative samples, AHNS$_{p=-1}$ and AHNS$_{p=-2}$ significantly outperform RNS by on average 20\%. Meanwhile, the two methods also show a huge performance boost over other hard negative sampling methods and the curriculum-learning-based method.
\end{itemize}

\begin{figure}[t]
  \centering
  \subfigure[Phones, Avg. Hardness]{\includegraphics[width=0.43\linewidth]{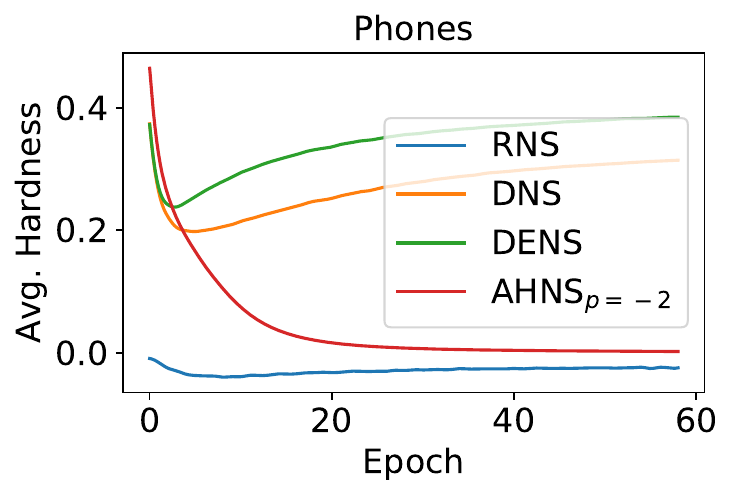} \label{fig:ph}}
  \subfigure[Phones, NDCG@20]{\includegraphics[width=0.40\linewidth]{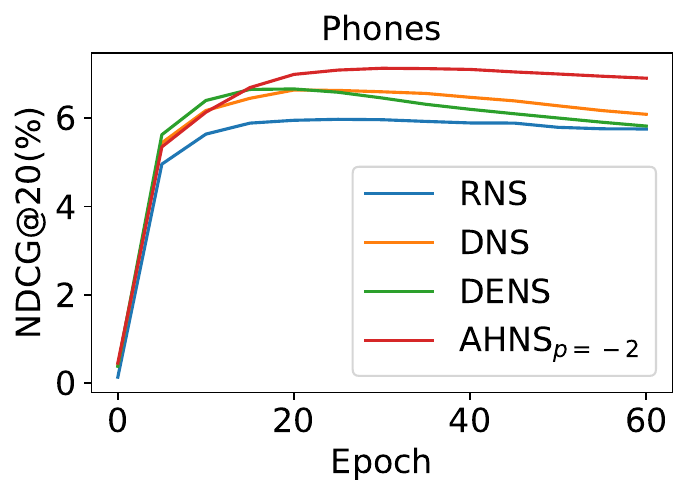} \label{fig:pn}}
  
  \subfigure[Sports, Avg. Hardness]{\includegraphics[width=0.43\linewidth]{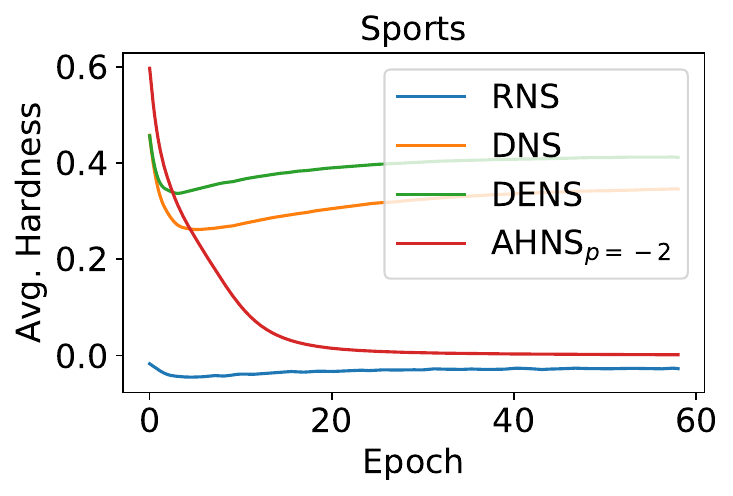} \label{fig:sh}}
  \subfigure[Sports, NDCG@20]{\includegraphics[width=0.40\linewidth]{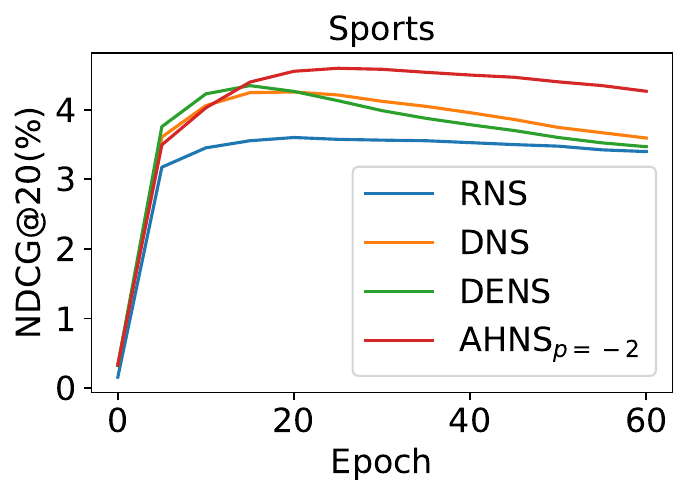} \label{fig:sn}}

  \subfigure[Tools, Avg. Hardness]{\includegraphics[width=0.43\linewidth]{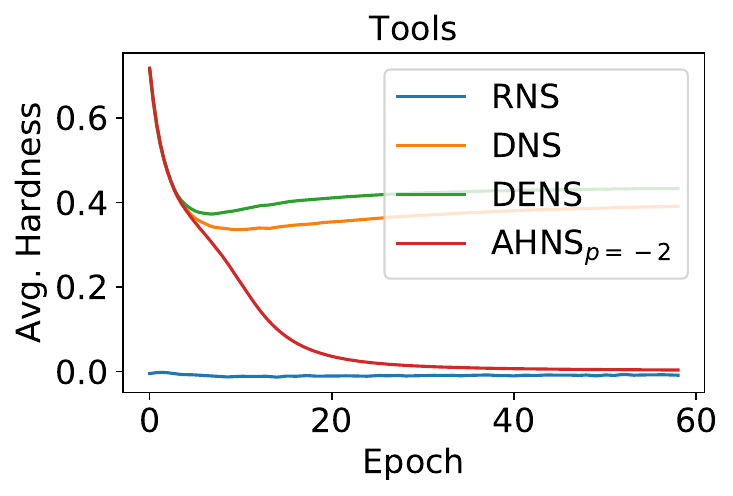} \label{fig:th}}
  \subfigure[Tools, NDCG@20]{\includegraphics[width=0.40\linewidth]{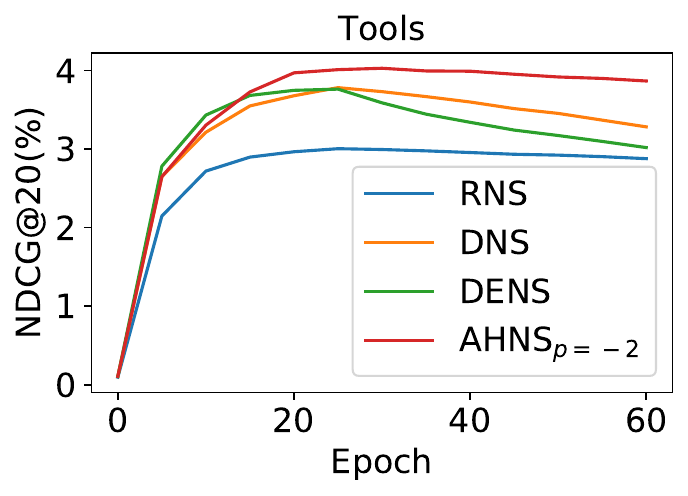} \label{fig:tn}}
  \caption{Average negative hardness and $NDCG$@20 of RNS, DNS, DENS, and AHNS$_{p=-2}$.}
  \label{fig:hard}
\end{figure}

\begin{figure}[t]
  \centering
  \subfigure[Phones, $\beta=0.1$]{\includegraphics[width=0.47\linewidth]{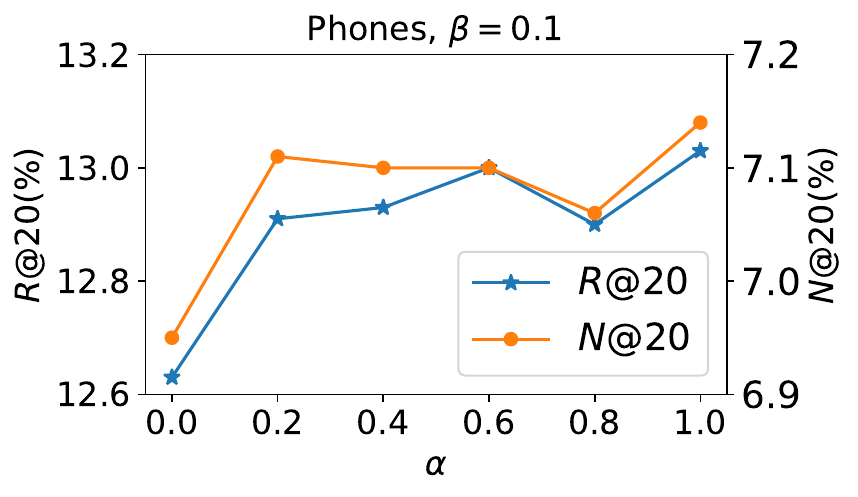}}
  \subfigure[Phones, $\alpha=1.0$]{\includegraphics[width=0.47\linewidth]{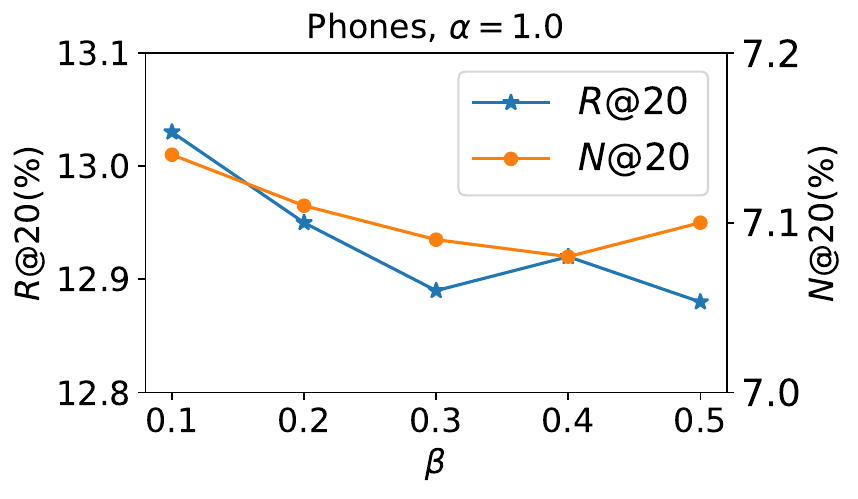}}
  
  \subfigure[Sports, $\beta=0.1$]{\includegraphics[width=0.47\linewidth]{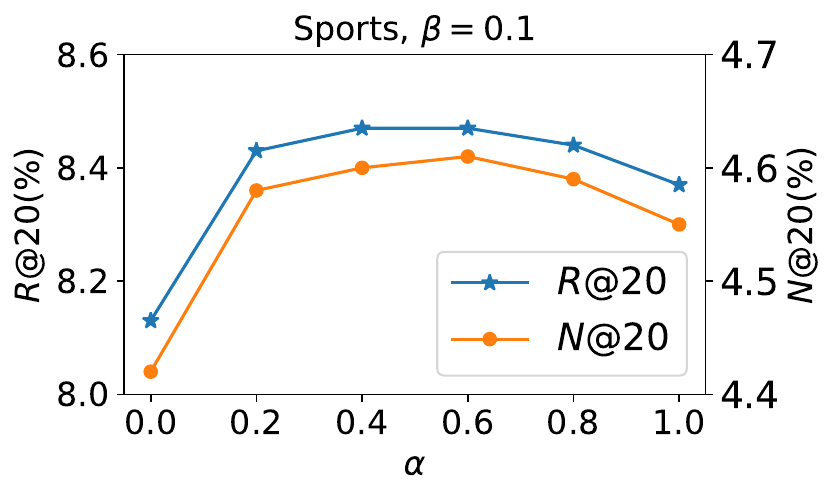}}
  \subfigure[Sports, $\alpha=1.0$]{\includegraphics[width=0.47\linewidth]{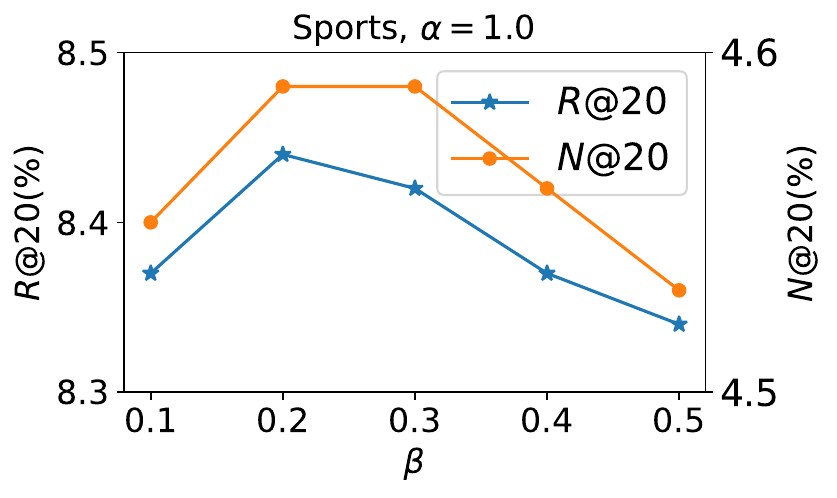}}

  \subfigure[Tools, $\beta=0.1$]{\includegraphics[width=0.47\linewidth]{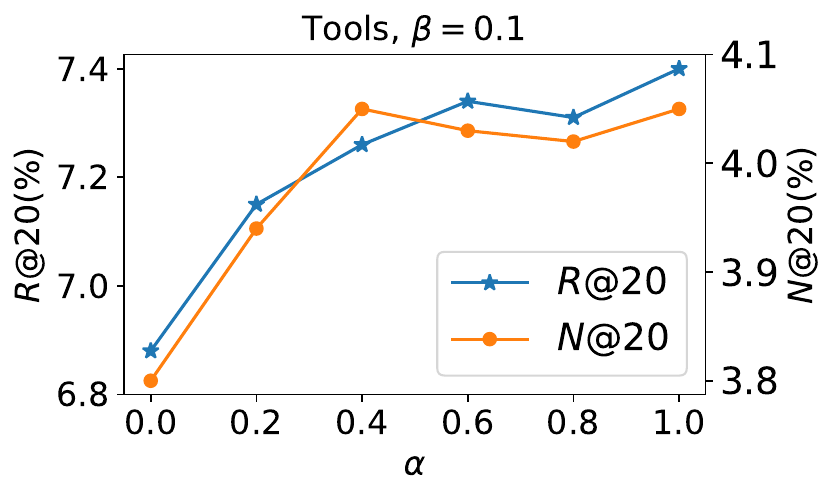}}
  \subfigure[Tools, $\alpha=1.0$]{\includegraphics[width=0.47\linewidth]{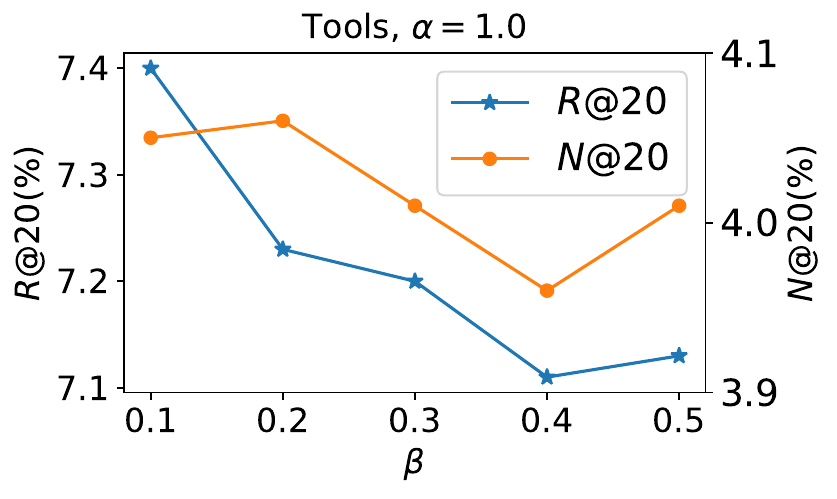}}
  \caption{Performance of AHNS$_{p=-2}$ \textit{w.r.t.} different hyperparameters.}
  \label{fig:hyp}
\end{figure}

\subsection{RQ2: Hardness Visualization}
\label{sec:rq2}

To justify the motivation of AHNS, i.e., adaptively selecting hardnesses of negative samples to alleviate FPP and FNP, we plot the curves of average negative hardness and $NDCG$@20 of RNS, DNS, DENS, and AHNS$_{p=-2}$ in Fig.~\ref{fig:hard} on the three Amazon datasets. From these figures, we have the following key findings:
\begin{itemize}
    \item As shown in Fig.~\ref{fig:ph}, \ref{fig:sh} and \ref{fig:th}, compared to fixed hardness negative sampling methods RNS, DNS, and DENS, AHNS$_{p=-2}$ can adaptively adjust the hardnesses of negative samples as training progresses. Specifically, in the early stages of training, AHNS$_{p=-2}$ favors negative samples with higher hardnesses, while in the later stages of training, AHNS$_{p=-2}$ prefers negative samples with lower hardnesses.
    \item As shown in Fig.~\ref{fig:pn}, \ref{fig:sn} and \ref{fig:tn}, the performance of RNS peaks in the early stages of training and remains stable thereafter, DNS and DENS present better performance than RNS but suffer a significant performance drop in the later stages of training, and AHNS$_{p=-2}$ achieves the best performance while maintaining similar stability as RNS.
    \item The average negative hardness and $NDCG$@20 of RNS (the blue line) well verify the existence of FPP, i.e., when only easy negative samples can be selected during the training process, items of no interest but with initially high predicted scores may not be sufficiently updated and will still be recommended to users, leading to the suboptimal performance of RNS. 
    \item The average negative hardness and $NDCG$@20 of DNS (the orange line) and DENS (the green line) well verify the existence of FNP, i.e., when only hard negative samples can be selected during the training process, items of interests may be selected as negative and ranked lower in the recommendation list, resulting in the performance drop of DNS and DENS.
    \item The average negative hardness and $NDCG$@20 of AHNS$_{p=-2}$ (the red line) well justify our motivation. For positives with lower predicted scores, by selecting items with higher hardnesses as negative, AHNS$_{p=-2}$ well alleviates FPP and achieves a higher peak; for positives with higher predicted scores, by selecting items with lower hardnesses as negative, AHNS$_{p=-2}$ well avoids FNP and thus prevents the performance drop.
\end{itemize}

\subsection{RQ3: Hyperparameter Study}

As discussed in Thm.~\ref{thm:c3}, hyperparameters $\alpha$ and $\beta$ affect the hardnesses of selected negative samples. Here we study how these hyperparameters affect the recommendation performance. Fig.~\ref{fig:hyp} shows $Recall$@20 and $NDCG$@20 of AHNS$_{p=-2}$ under different $\alpha$ or $\beta$ values with other hyperparameters unchanged on the three Amazon datasets. 
We can see that it is intractable to identify the optimal values of $\alpha$ and $\beta$ since they are different across datasets and evaluation metrics. However, in practice, we can achieve desirable performance in a relatively wide range of $\alpha$ or $\beta$ values, 
which relieves the overhead of hyperparameter tuning.

\begin{figure}[t]
  \centering
  \includegraphics[width=0.72\linewidth]{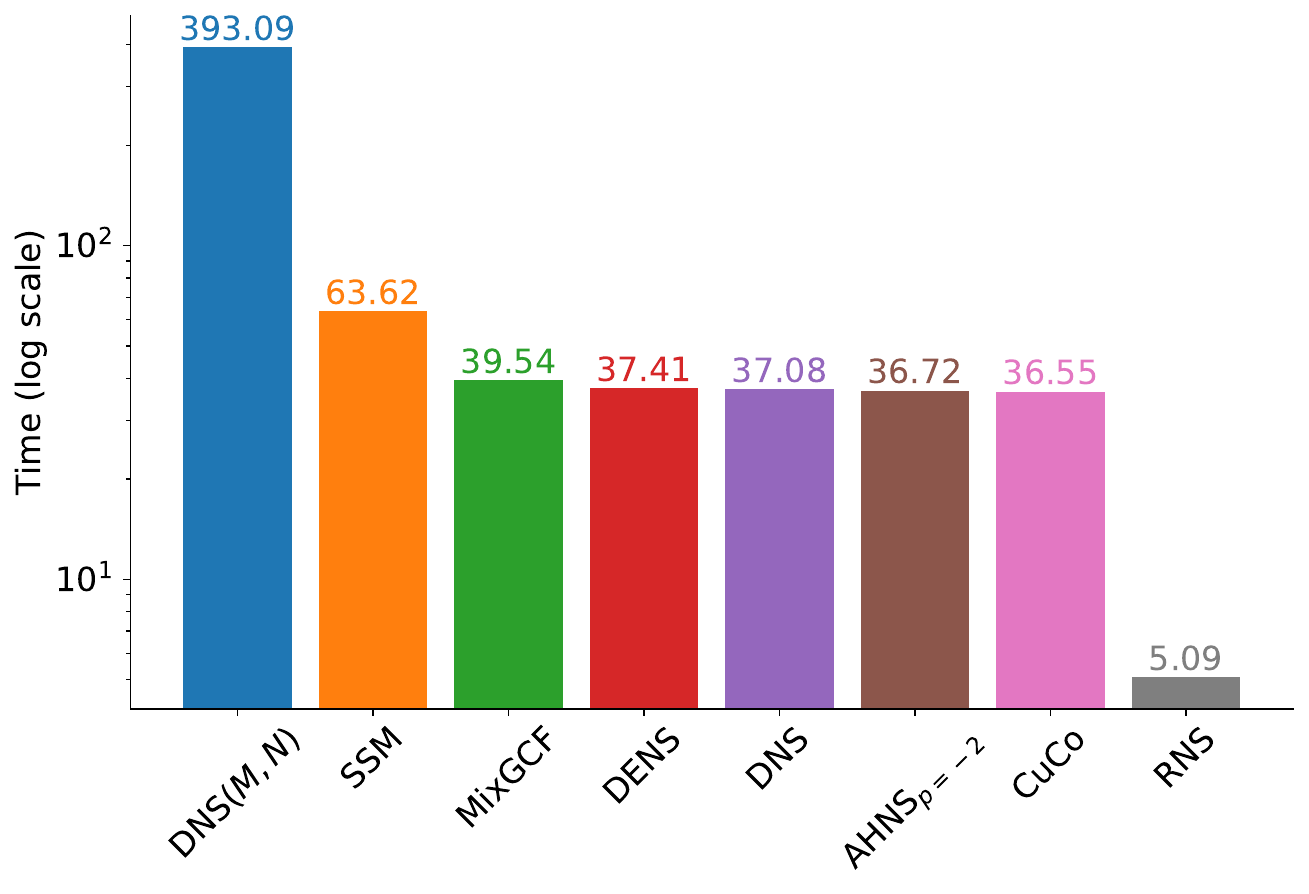}
  \caption{Time (second) for training per epoch on ML-1M \textit{w.r.t.} different methods. Best viewed in color.}
  \label{fig:time}
\end{figure}

\subsection{RQ4: Efficiency Analysis}

As presented in Algo.~\ref{alg:hans}, AHNS$_{p<0}$ does not introduce additional time cost compared to the simplest hard negative sampling method DNS. Here we empirically compare the time for training each epoch of AHNS$_{p=-2}$ and other baseline methods on the ML-1M dataset. All the methods are implemented under the same framework and with optimal hyperparameters to ensure fairness. The results are shown in Fig.~\ref{fig:time}. DNS($M, N$) takes the longest time as it requires an extremely large candidate negative set to adjust the hardness of negative samples. SSM costs the second longest time because multiple negative samples are selected to participate in the training of the CF model. The time difference between AHNS$_{p=-2}$ and other hard negative sampling methods is marginal, and RNS undoubtedly takes the least time. Considering the performance improvements in Tab.~\ref{tab:per} that AHNS$_{p=-2}$ can bring, we believe that AHNS$_{p=-2}$ is the best negative sampling method in terms of both efficiency and performance.

\section{Conclusion}
In this paper, we propose a new negative sampling paradigm AHNS with three key criteria, which enables adaptive selection of hardnesses of negative samples to alleviate FPP and FNP. We devise a concrete instantiation AHNS$_{p<0}$ and theoretically demonstrate that it can well fit the three criteria of AHNS and achieve a larger lower bound of NDCG. Comprehensive experiments confirm that AHNS$_{p<0}$ provides a promising new research direction for negative sampling to further boost implicit CF models' performance.

\clearpage

\section{Acknowledgments}
This work was supported by the Heilongjiang Key R\&D Program of China under Grant No. GA23A915 and the National Natural Science Foundation of China under Grant No. 62072136. It was also partially supported by Hong Kong Baptist University IG-FNRA project under Grant No. RC-FNRA-IG/21-22/SCI/01.

\bibliography{aaai24}

\begin{thebibliography}{24}
\providecommand{\natexlab}[1]{#1}

\bibitem[{Cai and Wang(2018)}]{CW18}
Cai, L.; and Wang, W.~Y. 2018.
\newblock {KBGAN:} Adversarial Learning for Knowledge Graph Embeddings.
\newblock In \emph{Proceedings of the 2018 Conference of the North American Chapter of the Association for Computational Linguistics: Human Language Technologies}, 1470--1480.

\bibitem[{Chen et~al.(2023)Chen, Dong, Wang, Feng, Wang, and He}]{CDW23}
Chen, J.; Dong, H.; Wang, X.; Feng, F.; Wang, M.; and He, X. 2023.
\newblock Bias and Debias in Recommender System: A Survey and Future Directions.
\newblock \emph{ACM Transactions on Information Systems}, 41(3): 1--39.

\bibitem[{Chen et~al.(2022)Chen, Lian, Jin, Zheng, and Chen}]{CLJ22}
Chen, J.; Lian, D.; Jin, B.; Zheng, K.; and Chen, E. 2022.
\newblock Learning Recommenders for Implicit Feedback with Importance Resampling.
\newblock In \emph{Proceedings of the ACM Web Conference 2022}, 1997--2005.

\bibitem[{Chen et~al.(2017)Chen, Sun, Shi, and Hong}]{CSS17}
Chen, T.; Sun, Y.; Shi, Y.; and Hong, L. 2017.
\newblock On Sampling Strategies for Neural Network-based Collaborative Filtering.
\newblock In \emph{Proceedings of the 23rd ACM SIGKDD Conference on Knowledge Discovery and Data Mining}, 767--776.

\bibitem[{Chu et~al.(2021)Chu, Wang, Shi, and Jiang}]{CWS21}
Chu, G.; Wang, X.; Shi, C.; and Jiang, X. 2021.
\newblock CuCo: Graph Representation with Curriculum Contrastive Learning.
\newblock In \emph{Proceedings of the 30th International Joint Conference on Artificial Intelligence}, 2300--2306.

\bibitem[{Ding et~al.(2020)Ding, Quan, Yao, Li, and Jin}]{DQY20}
Ding, J.; Quan, Y.; Yao, Q.; Li, Y.; and Jin, D. 2020.
\newblock Simplify and Robustify Negative Sampling for Implicit Collaborative Filtering.
\newblock In \emph{Proceedings of the 34th International Conference on Neural Information Processing Systems}.

\bibitem[{He et~al.(2020)He, Deng, Wang, Li, Zhang, and Wang}]{HDK20}
He, X.; Deng, K.; Wang, X.; Li, Y.; Zhang, Y.; and Wang, M. 2020.
\newblock LightGCN: Simplifying and Powering Graph Convolution Network for Recommendation.
\newblock In \emph{Proceedings of the 43rd International ACM SIGIR Conference on Research and Development in Information Retrieval}, 639--648.

\bibitem[{Huang et~al.(2021)Huang, Dong, Ding, Yang, Feng, Wang, and Tang}]{HDD21}
Huang, T.; Dong, Y.; Ding, M.; Yang, Z.; Feng, W.; Wang, X.; and Tang, J. 2021.
\newblock MixGCF: An Improved Training Method for Graph Neural Network-based Recommender Systems.
\newblock In \emph{Proceedings of the 27th ACM SIGKDD Conference on Knowledge Discovery and Data Mining}, 665--674.

\bibitem[{Kingma and Ba(2015)}]{KB15}
Kingma, D.~P.; and Ba, J. 2015.
\newblock Adam: A Method for Stochastic Optimization.
\newblock In \emph{Proceedings of the 3rd International Conference on Learning Representations}.

\bibitem[{Lai et~al.(2023)Lai, Chen, Zhao, Chen, and Han}]{LCZ23}
Lai, R.; Chen, L.; Zhao, Y.; Chen, R.; and Han, Q. 2023.
\newblock Disentangled Negative Sampling for Collaborative Filtering.
\newblock In \emph{Proceedings of the 16th International Conference on Web Search And Data Mining}, 96--104.

\bibitem[{Lin et~al.(2017)Lin, Goyal, Girshick, He, and Doll{\'a}r}]{LGG17}
Lin, T.-Y.; Goyal, P.; Girshick, R.; He, K.; and Doll{\'a}r, P. 2017.
\newblock Focal Loss for Dense Object Detection.
\newblock In \emph{Proceedings of the IEEE International Conference on Computer Vision}, 2980--2988.

\bibitem[{Mikolov et~al.(2013)Mikolov, Sutskever, Chen, Corrado, and Dean}]{MSC13}
Mikolov, T.; Sutskever, I.; Chen, K.; Corrado, G.~S.; and Dean, J. 2013.
\newblock Distributed Representations of Words and Phrases and their Compositionality.
\newblock In \emph{Proceedings of the 27th International Conference on Neural Information Processing Systems}, 3111--3119.

\bibitem[{Park and Chang(2019)}]{PC19}
Park, D.~H.; and Chang, Y. 2019.
\newblock Adversarial Sampling and Training for Semi-Supervised Information Retrieval.
\newblock In \emph{Proceedings of the 28th International Conference on World Wide Web}, 1443--1453.

\bibitem[{Rendle et~al.(2009)Rendle, Freudenthaler, Gantner, and Schmidt-Thieme}]{RFG09}
Rendle, S.; Freudenthaler, C.; Gantner, Z.; and Schmidt-Thieme, L. 2009.
\newblock BPR: Bayesian Personalized Ranking from Implicit Feedback.
\newblock In \emph{Proceedings of the 25th Conference on Uncertainty in Artificial Intelligence}, 452--461.

\bibitem[{Shi et~al.(2023)Shi, Chen, Feng, Zhang, Wu, Gao, and He}]{SCF23}
Shi, W.; Chen, J.; Feng, F.; Zhang, J.; Wu, J.; Gao, C.; and He, X. 2023.
\newblock On the Theories Behind Hard Negative Sampling for Recommendation.
\newblock In \emph{Proceedings of the ACM Web Conference 2023}, 812--822.

\bibitem[{Wang et~al.(2017)Wang, Yu, Zhang, Gong, Xu, Wang, Zhang, and Zhang}]{WYZ17}
Wang, J.; Yu, L.; Zhang, W.; Gong, Y.; Xu, Y.; Wang, B.; Zhang, P.; and Zhang, D. 2017.
\newblock IRGAN: A Minimax Game for Unifying Generative and Discriminative Information Retrieval Models.
\newblock In \emph{Proceedings of the 40rd International ACM SIGIR Conference on Research and Development in Information Retrieval}, 515--524.

\bibitem[{Wang et~al.(2019)Wang, He, Wang, Feng, and Chua}]{WHW19}
Wang, X.; He, X.; Wang, M.; Feng, F.; and Chua, T.-S. 2019.
\newblock Neural Graph Collaborative Filtering.
\newblock In \emph{Proceedings of the 42nd International ACM SIGIR Conference on Research and Development in Information Retrieval}, 165--174.

\bibitem[{Wu et~al.(2019)Wu, Volkovs, Soon, Sanner, and Rai}]{WVS19}
Wu, G.; Volkovs, M.; Soon, C.~L.; Sanner, S.; and Rai, H. 2019.
\newblock Noise Contrastive Estimation for One-Class Collaborative Filtering.
\newblock In \emph{Proceedings of the 42nd International ACM SIGIR Conference on Research and Development in Information Retrieval}, 135--144.

\bibitem[{Wu et~al.(2022)Wu, Wang, Gao, Chen, Fu, Qiu, and He}]{WWG22}
Wu, J.; Wang, X.; Gao, X.; Chen, J.; Fu, H.; Qiu, T.; and He, X. 2022.
\newblock On the Effectiveness of Sampled Softmax Loss for Item Recommendation.
\newblock \emph{arXiv preprint arXiv:2201.02327}.

\bibitem[{Xu et~al.(2022)Xu, Lian, Zhao, Gong, Shou, Jiang, Xie, and Wen}]{XLZ22}
Xu, L.; Lian, J.; Zhao, W.~X.; Gong, M.; Shou, L.; Jiang, D.; Xie, X.; and Wen, J.-R. 2022.
\newblock Negative Sampling for Contrastive Representation Learning: A Review.
\newblock \emph{arXiv preprint arXiv:2206.00212}.

\bibitem[{Yang et~al.(2020)Yang, Ding, Zhou, Yang, Zhou, and Tang}]{YDZ20}
Yang, Z.; Ding, M.; Zhou, C.; Yang, H.; Zhou, J.; and Tang, J. 2020.
\newblock Understanding Negative Sampling in Graph Representation Learning.
\newblock In \emph{Proceedings of the 26th ACM SIGKDD Conference on Knowledge Discovery and Data Mining}, 1666--1676.

\bibitem[{Zhang et~al.(2013)Zhang, Chen, Wang, and Yu}]{ZCW13}
Zhang, W.; Chen, T.; Wang, J.; and Yu, Y. 2013.
\newblock Optimizing Top-N Collaborative Filtering via Dynamic Negative Item Sampling.
\newblock In \emph{Proceedings of the 36th International ACM SIGIR Conference on Research and Development in Information Retrieval}, 785--788.

\bibitem[{Zhao et~al.(2023)Zhao, Chen, Lai, Han, Song, and Chen}]{ZCL23}
Zhao, Y.; Chen, R.; Lai, R.; Han, Q.; Song, H.; and Chen, L. 2023.
\newblock Augmented Negative Sampling for Collaborative Filtering.
\newblock In \emph{Proceedings of the 17th ACM Conference on Recommender Systems}.

\bibitem[{Zhu et~al.(2022)Zhu, Zhang, He, and Dou}]{ZZH22}
Zhu, Q.; Zhang, H.; He, Q.; and Dou, Z. 2022.
\newblock A Gain-Tuning Dynamic Negative Sampler for Recommendation.
\newblock In \emph{Proceedings of the Web Conference 2022}, 277--285.

\end{thebibliography}

\end{document}